\documentclass[11pt,oneside]{amsart}
\def\isdraft{0}

\usepackage{mathtools}
\usepackage{amsmath,amsfonts,amsthm,amssymb}
\usepackage[dvipsnames]{xcolor}
\usepackage{graphicx}
\usepackage{amsaddr} 
\usepackage{prettyref} 
\usepackage[utf8]{inputenc}
\usepackage{enumitem}
\usepackage[hyphens]{url} 
\usepackage{tikz-cd}
\usepackage{tikz}
\usetikzlibrary{positioning}
\usetikzlibrary{arrows}
\usepackage{hyperref}
\usepackage{a4wide}
\usepackage[color=black,textcolor=white\if\isdraft0,disable\fi]{todonotes}

\usepackage{bbm} 
\usepackage{bussproofs} 
\usepackage[mathscr]{euscript} 
\usepackage{stackengine} 
\usepackage{stmaryrd} 
\usepackage{wrapfig} 
\usepackage{marginnote} 
\usepackage{float} 
\usepackage{footnote} 
\usepackage{tabularx} 
\usepackage[normalem]{ulem} 
\usepackage{pifont} 
\usepackage[most]{tcolorbox} 
\usepackage{shuffle} 

\graphicspath{{Figures/}}

\newtheorem{theorem}{Theorem}

\newtheorem{conjecture}[theorem]{Conjecture}
\newtheorem{corollary}[theorem]{Corollary}
\newtheorem{fact}[theorem]{Fact}
\newtheorem{lemma}[theorem]{Lemma}
\newtheorem{proposition}[theorem]{Proposition}

\theoremstyle{definition} 

\newtheorem{definition}[theorem]{Definition}

\newtheorem{example}[theorem]{Example}

\newtheorem{remark}[theorem]{Remark}

\newrefformat{sec}{§\ref{#1}}
\newrefformat{a}{Answer \ref{#1}}
\newrefformat{alg}{Algorithm \ref{#1}}
\newrefformat{ax}{Appendix \ref{#1}}
\newrefformat{cl}{Claim \ref{#1}}
\newrefformat{con}{Conclusion \ref{#1}}
\newrefformat{conv}{Convention \ref{#1}}
\newrefformat{cj}{Conjecture \ref{#1}}
\newrefformat{c}{Corollary \ref{#1}}
\newrefformat{q}{(\ref{#1})}
\newrefformat{e}{Example \ref{#1}}
\newrefformat{exe}{Exercise \ref{#1}}
\newrefformat{d}{Definition \ref{#1}}
\newrefformat{Done}{Done \ref{#1}}
\newrefformat{f}{Fact \ref{#1}}
\newrefformat{fig}{Figure \ref{#1}}
\newrefformat{h}{Hypothesis \ref{#1}}
\newrefformat{idea}{Idea \ref{#1}}
\newrefformat{i}{Item \ref{#1}}
\newrefformat{l}{Lemma \ref{#1}}
\newrefformat{N}{Note \ref{#1}}
\newrefformat{n}{Notation \ref{#1}}
\newrefformat{o}{Observation \ref{#1}}
\newrefformat{pr}{Problem \ref{#1}}
\newrefformat{p}{Proposition \ref{#1}}
\newrefformat{pc}{Pseudocode \ref{#1}}
\newrefformat{Q}{Q. \ref{#1}}
\newrefformat{r}{Remark \ref{#1}}
\newrefformat{tab}{Table \ref{#1}}
\newrefformat{t}{Theorem \ref{#1}}
\newrefformat{T}{TODO \ref{#1}}
\newrefformat{traum}{Traum \ref{#1}}
\newrefformat{w}{Warning \ref{#1}}

\newcommand{\boxblacktriangle}{\mathrel{\ooalign{$\square$\cr\kern0.07ex\hbox{\scalebox{0.9}{$\blacktriangle$}}}}}

\newcommand{\boxtriangle}{\mathrel{\ooalign{$\square$\cr\kern0.07ex\hbox{\scalebox{0.9}{$\triangle$}}}}}

\newlist{todolist}{itemize}{2}
\setlist[todolist]{label=$\square$}

\usepackage[top=0.7cm,bottom=0.5cm,left=1cm,right=1cm]{geometry} 

\if\isdraft1
\fi

\title{
	The algebra of Krom logic programs
}
\author{
	Christian Anti\'c
}
\address{
	christian.antic@icloud.com\\
	Vienna University of Technology\\
	Vienna, Austria
}

\begin{document}

\begin{abstract}
This paper investigates the algebraic structure of Krom logic programs, consisting only of facts and rules with at most one body atom. We show that sequential composition endows the class of Krom programs with a natural monoid structure and that this structure admits rich algebraic extensions to Krom seminearrings, Krom quemirings, Krom-Conway seminearrings, and Krom-Conway omegaseminearrings. Furthermore, we establish explicit generating sets and canonical decompositions, study the associated ${}^\omega$-operator, characterize the Kleene star in graph-theoretic terms, and relate finite Krom monoids to transformation monoids and finite-state automata. These results provide new connections between logic programming, algebraic automata theory, and algebraic graph theory.
\end{abstract}

\maketitle


\tableofcontents

\section{Introduction}

Logic programs constitute one of the most extensively studied rule-based formalisms in theoretical computer science and mathematical logic (cf.~\cite{Apt90,Lloyd87}). Among them, Horn theories occupy a particularly prominent position due to their favorable computational and logical properties and their numerous applications in computer science~\cite{Makowsky87}.

The sequential composition of propositional logic programs was introduced in~\cite{Antic21-1,Antic23-23} and subsequently extended to answer set programming in~\cite{Antic21-2}. This operation endows logic programs with an algebraic structure that has so far received comparatively little attention.

More broadly, this work is part of an ongoing effort to develop an algebraic theory of logic programming, in which logic programs are studied as algebraic objects equipped with natural operations and identities. From this perspective, sequential composition plays a role analogous to multiplication in semiring theory and provides a foundation for investigating structural decomposition, generation, and transformation of logic programs.

The Krom fragment \cite{Krom67} occupies a distinguished position within this program. On the one hand, it is expressive enough to exhibit non-trivial algebraic phenomena, including natural connections with graph theory, permutation groups, and automata. On the other hand, it enjoys structural properties that fail for general logic programs, such as the compatibility of sequential composition with set union and the proper rule operator. Moreover, Krom programs arise as canonical building blocks in decomposition results for Horn programs~\cite{Antic24-1}. These observations suggest that the Krom fragment provides a natural starting point for the development of an algebraic theory of logic programming.

The main objective of this paper is to identify and study the algebraic structures naturally induced by sequential composition on Krom logic programs. We show that this seemingly elementary fragment supports a surprisingly rich hierarchy of algebraic structures, including monoids, quemirings, Conway seminearrings, and omegaseminearrings, and admits close connections with graph theory, transformation monoids, and automata.

Besides being of independent interest, Krom programs arise naturally as the basic building blocks in the algebraic theory of set-like operations on logic programs~\cite[Minimalist Decomposition Theorem]{Antic24-1}. In particular, every \emph{minimalist}\footnote{A propositional Horn program is called \emph{minimalist} in~\cite{Antic24-1} if it contains at most one rule for each rule head.} propositional Horn program admits a canonical decomposition into Krom components, providing further motivation for studying their algebraic properties.

A characteristic feature of the Krom fragment is that sequential composition distributes over union from the left, a property that fails for arbitrary propositional logic programs (cf.~\cite[Example~8]{Antic21-1}). Furthermore, the proper rule operator, which extracts the non-factual rules of a program, is an endomorphism with respect to sequential composition. These observations make the Krom fragment particularly amenable to an algebraic treatment.

The main contribution of this paper is the development of an algebraic theory of Krom logic programs under sequential composition. We show that the resulting structures naturally give rise to monoids, seminearrings, quemirings, Conway seminearrings, and Conway omegaseminearrings, thereby revealing close connections with semiring-based algebraic automata theory~\cite{Esik06} and algebraic graph theory~\cite{Knauer11}. In particular, sequential composition corresponds to relational composition, while the Kleene star captures graph reachability.

In a broader sense, this paper constitutes another step towards an algebraic theory of logic programming, a research direction initiated by Richard O'Keefe~\cite{OKeefe85} and subsequently pursued in the context of modular logic programming (see, e.g.,~\cite{Brogi99,Bugliesi94}) as well as in the author's recent work~\cite{Antic14,Antic21-2,Antic21-1,Antic23-23,Antic24-1}.

\section{Preliminaries}\label{sec:Pre}

In this section, we recall the syntax, semantics, and composition of (possibly infinite) propositional Krom\footnote{The Krom fragment is named after the author of \cite{Krom67}.} logic programs, consisting only of rules with at most a single body atom, by following the lines of \cite{Antic21-1}.


Let $A$ be an alphabet of propositional atoms. A \emph{(propositional Krom logic) program} over $A$ is a set of \emph{(Krom) rules} of two forms, 
(i) \emph{facts} of the form $a\in A$, and 
(ii) \emph{proper rules} of the form $a\leftarrow b$, for $a,b\in A$. We denote the set of all Krom programs over $A$ by $\mathbb K_A$ or simply by $\mathbb K$ in case $A$ is understood from the context. We denote the facts and proper rules in $K$ by $f(K)$ and $p(K)$, respectively.



An \emph{interpretation} is any subset of $A$. We define the \emph{entailment relation}, for every interpretation $I$, inductively as follows:
(1) for an atom $a$, $I\models a$ if $a\in I$;
(2) for a proper rule, $I\models a\leftarrow b$ if $I\models b$ implies $I\models a$;
(3) for a propositional Krom program $K$, $I\models K$ if $I\models r$ holds for each rule $r\in K$.

In case $I\models K$, we call $I$ a \emph{model} of $K$. The set of all models of $K$ has a least element with respect to set inclusion called the \emph{least model} of $K$ and denoted by $LM(K)$.

We shall now recall the sequential composition \cite{Antic21-1} of Krom programs. 

In the rest of the paper, $K$ and $L$ denote Krom programs over some fixed alphabet $A$. 


Define the (\emph{sequential}) \emph{composition} of $K$ and $L$ by
\begin{align*} 
	K\circ L := f(K)\cup\{a\in A\mid a\leftarrow b\in K,\; b\in L\}\cup\{a\leftarrow b\mid a\leftarrow c\in K,\; c\leftarrow b\in L\}.
\end{align*} We simply write $KL$ in case the composition operation is understood. 

We define the \emph{unit program} over $A$ by
\begin{align*} 
	1_A := \{a\leftarrow a\mid a\in A\}.
\end{align*} As usual, we often omit $A$ and write $1$.



Interpretations are left zeros, that is,
\begin{align}
	\label{eq:IK=I} IK = I
\end{align} holds for every Krom program $K$ and interpretation $I$.

We have the following structural result:

\begin{theorem}{{\cite[Thm. 12]{Antic21-1}}}\label{t:seminearring} The structure $(\mathbb K, \cup, \circ, \emptyset, 1)$ forms a seminearring \cite{vanHoorn67} (and almost\footnote{The inequality $K\emptyset\neq\emptyset$ is the only reason why the set of \emph{all} Krom programs (possibly containing facts) fails to be a semiring; but see \prettyref{t:p(mathbb_K)}.} a semiring) as it satisfies the following identities, for all Krom programs $K,L,M\in \mathbb K$:
\begin{align} 
	(K\cup L)\cup M &= K\cup (L\cup M)\\
	K\cup L &= L\cup K\\
	\emptyset\cup K &= K\cup\emptyset = K\\
	\label{eq:K(LM)=(KL)M} K(LM) &= (KL)M\\
	K1 &= 1K=K\\
	\label{eq:M(K_cup_L)} M(K\cup L) &= KK\cup ML\\
	\label{eq:(K_cup_L)M} (L\cup M)K &= KM\cup LK\\
	\emptyset K &= \emptyset.
\end{align} From now on, we call $(\mathbb K, \cup, \circ, \emptyset, 1)$ the \emph{Krom seminearring}.
\end{theorem}

\section{The proper rule operator}\label{sec:p}

In this section, we investigate the proper rule operator $p$, which extracts the proper rules of a Krom program by removing all facts. Besides serving as a basic structural decomposition, the operator $p$ exhibits remarkable algebraic properties with respect to sequential composition and set union. In particular, it gives rise to a natural quemiring structure on finite Krom seminearrings, thereby revealing an unexpected connection between Krom programs and algebraic structures originating in automata theory.

\subsection{Basic properties}

We begin by collecting the fundamental algebraic properties of the proper rule operator. In particular, we show that $p$ is compatible with the basic operations on Krom programs and therefore behaves as a natural algebraic projection onto the proper part of a program. These identities will serve as the foundation for the richer structures developed in the subsequent subsection.

\begin{theorem}\label{t:p_hom} The proper rule operator $p$ is an endomorphism of the Krom seminearring $(\mathbb K, \cup, \circ, \emptyset, 1)$ thus satisfying
\begin{align} 
	\label{eq:p(K_cup_L)} p(K\cup L) &= p(K)\cup p(L),\\
	\label{eq:p(KL)} p(KL) &= p(K)p(L),\\
	\label{eq:p(1)} p(1) &= 1,\\
	\label{eq:p(0)} p(\emptyset) &= \emptyset,
\end{align} for all Krom programs $K$ and $L$.
\end{theorem}
\begin{proof} An immediate consequence from the definitions.
\end{proof}

\begin{remark} A distinguishing feature of the Krom fragment is that the proper rule operator commutes with sequential composition. In contrast, this identity does not hold for arbitrary propositional logic programs; see \cite[§4.4]{Antic21-1}.
\end{remark}

\begin{theorem}\label{t:p(mathbb_K)} 
The Krom algebra $(p(\mathbb K), \cup, \circ, \emptyset, 1)$ containing only proper Krom programs is an idempotent semiring.
\end{theorem}
\begin{proof} For every proper Krom program $K\in p( \mathbb K)$, we have
\begin{align*} 
	K\emptyset = \emptyset.
\end{align*} Hence the right annihilation axiom holds on $p(\mathbb K)$, so together with \prettyref{t:seminearring} all semiring axioms are satisfied.
\end{proof}

Thus, removing all facts restores the missing right-annihilation property and turns the Krom seminearring into a semiring.

\subsection{Krom quemirings}

Quemirings were introduced in \cite{Elgot76} in the context of algebraic automata theory (see e.g. \cite[p.110]{Esik06}). The following theorem shows that they also arise naturally in the algebra of Krom logic programs:

\begin{theorem}
The Krom algebra $(\mathbb K, \cup, \circ, \emptyset, 1, p)$ forms a quemiring, henceforth called the \emph{Krom quemiring}. Explicitly, for all $K,L,M\in\mathbb K$, the following identities hold in addition to the identities in \prettyref{t:seminearring}:
\begin{align*}
	p(K)(L\cup M) &= (p(K)L)\cup(p(K)M),\\
	K &= p(K)\cup(K \emptyset),\\
	p(K) \emptyset &= \emptyset,\\
	p(K\cup L) &= p(K)\cup p(L),\\
	p(KL) &= p(K)p(L).
\end{align*}
\end{theorem}
\begin{proof} The first identity holds by distributivity \prettyref{eq:M(K_cup_L)}. The second and third identities follow from the fact that $K \emptyset \stackrel{\ref{eq:f(K)=K0}}= f(K)$. The last two identities hold trivially.
\end{proof}

\section{The facts operator}\label{sec:facts}

This section studies elementary algebraic properties of the facts operator $f$, which extracts the facts of a Krom program.

Our first observation is that the facts operator can be reduced to composition via
\begin{align} 
	\label{eq:f(K)=K0} f(K) = K\emptyset.
\end{align} This means that we can extract the facts of a Krom program $K$ by composing it with the empty program.

The next result summarizes the interaction of the facts operator with set union and intersection:

\begin{proposition} The facts operator is a homomorphism of bounded lattices $$f : ( \mathbb K_A, \cup, \cap, \emptyset, F_A)\to ( \mathbb I_A, \cup, \cap, \emptyset, A ),$$ where $F_A$ denotes the Krom program containing all Krom rules over $A$. In other words,
\begin{align} 
	\label{eq:f(K_cup_L)} f(K\cup L) &= f(K)\cup f(L)\\
	f(K\cap L) &= f(K)\cap f(L)\\
	f(\emptyset) &= \emptyset\\
	f(F_A) &= A.
\end{align}
\end{proposition}

We now want to study the interaction between the facts operator and composition. The facts of the composition $KL$ can be computed as
\begin{align} 
	\label{eq:f(KL)} f(KL) &= f(K)\cup p(K)f(L),
\end{align} where $p(K)$ denotes the proper rules in $K$.

More generally, the next result provides a concise formula for the computation of the facts of a sequence of program compositions which will often be used below:

\begin{theorem} We can compute the facts of the sequential composition of the Krom programs $K_1, \ldots, K_n$, $n\geq 1$, by
\begin{align}
	\label{eq:f(K1-Kn)}  f(K_1\ldots K_n) = f(K_1)\cup\bigcup_{i=1}^{n-1}(p(K_1\ldots K_i)f(K_{i+1})).
\end{align}
\end{theorem}
\begin{proof} By induction on $n$. The induction base $n=1$ holds trivially. For the induction step, we compute
\begin{align*} 
	f(K_1\ldots K_{n+1})
		&\stackrel{\ref{eq:f(KL)}}= f(K_1)\cup p(K_1)f(K_2\ldots K_{n+1})\\
		&\stackrel{IH}= f(K_1)\cup p(K_1)\left(f(K_2)\cup \bigcup_{i=2}^{n}(p(K_2\ldots K_i)f(K_{i+1}))\right)\\
		&\stackrel{\ref{eq:M(K_cup_L)}}= f(K_1)\cup p(K_1)f(K_2)\cup \bigcup_{i=2}^{n}(p(K_1)p(K_2\ldots K_i)f(K_{i+1}))\\
		&\stackrel{\ref{eq:p(KL)}}=f(K_1)\cup p(K_1)f(K_2)\cup \bigcup_{i=2}^{n}(p(K_1K_2\ldots K_i)f(K_{i+1}))\\
		&= f(K_1)\cup\bigcup_{i=1}^n(p(K_1\ldots K_i)f(K_{i+1})).
\end{align*}
\end{proof}

\section{Krom-Conway seminearrings} \label{sec:Conway}

Krom programs can naturally be viewed as directed graphs with designated vertices, with sequential composition corresponding to relational composition and the Kleene star corresponding to reachability.

More precisely, every proper Krom program $K$ over an alphabet $A$ can be identified with the directed graph whose vertex set is $A$ and whose edge set consists of all pairs $(b,a)$ such that the rule $a \leftarrow b$ belongs to $K$. Under this correspondence, sequential composition coincides with relational composition of directed graphs.

This interpretation motivates the introduction of the \emph{Kleene star} and \emph{Kleene plus} on proper Krom programs. For every proper Krom program $K$, we define
\[
K^\ast := \bigcup_{n\geq 0}K^n \quad\text{and}\quad K^+ := KK^\ast.
\]
Intuitively, $K^\ast$ contains precisely those rules $a\leftarrow b$ for which there exists a directed path from $b$ to $a$ in the graph represented by the proper rules of $K$. Thus, $K^\ast$ can be viewed as the reflexive transitive closure of $K$ and provides an algebraic characterization of reachability.

We now investigate the interaction of the Kleene star with the remaining algebraic operations.

\begin{proposition} For any Krom program $K$ and interpretation $I$,
\begin{align} 
	\label{eq:K_cup_I=I^astK} K\cup I = I^\ast K.
\end{align}
\end{proposition}
\begin{proof} The Kleene star acts on interpretations as
\begin{align} 
	\label{eq:I^ast} I^\ast = 1\cup I.
\end{align} This yields
\begin{align*} 
	K\cup I \stackrel{\ref{eq:IK=I}}= K\cup IK \stackrel{\ref{eq:(K_cup_L)M}}= (1\cup I)K \stackrel{\ref{eq:I^ast}}= I^\ast K.
\end{align*}
\end{proof}

\begin{proposition} The Kleene star and the proper rule operator $p$ are compatible in the sense that
\begin{align} 
	\label{eq:p(K^ast)} p(K^\ast) = p(K)^\ast
\end{align} holds for every Krom program $K$.
\end{proposition}
\begin{proof}
\begin{align*} 
	p(K^\ast) = p(\bigcup_{n\geq 0} K^n) \stackrel{\ref{eq:p(K_cup_L)}}= \bigcup_{n\geq 0} p(K^n) \stackrel{\ref{eq:p(KL)}}= \bigcup_{n\geq 0} p(K)^n = p(K)^\ast.
\end{align*}
\end{proof}



In the algebraic theory of formal languages and automata (see, e.g., \cite{Esik06}), the so-called \emph{Conway axioms}, introduced by J. H. Conway in \cite{Conway71}, play a fundamental role. We briefly recall two of them here.

Let $(S,+,\cdot,0,1,{}^\ast)$ be a starsemiring in the sense of \cite{Esik06}, that is, a semiring $(S,+,\cdot,0,1)$ enriched with an unary star operation ${}^\ast$. Two equations from \cite[p.15]{Conway71} will be important for our purposes (see also \cite[p.15]{Esik06}):
\begin{itemize}
\item The \emph{sum-star equation}
\begin{align*}
	(a+b)^\ast=(a^\ast b)^\ast a^\ast
\end{align*}
for all $a,b\in S$.

\item The \emph{product-star equation}
\begin{align*}
	(ab)^\ast=1+a(ba)^\ast b
\end{align*}
for all $a,b\in S$.
\end{itemize}

Conway identities provide an algebraic axiomatization of Kleene star and play a central role in the theory of regular languages and automata. It is therefore natural to ask whether the Kleene star arising from Krom programs satisfies these identities.

There is a natural analogy between formal languages and Krom programs, and more generally between starsemirings and starseminearrings. This motivates the following definition:

\begin{definition}
A \emph{Conway seminearring} is a seminearring enriched with a star operation satisfying the sum-star equation and the product-star equation.
\end{definition}

\begin{theorem}\label{t:Conway} The Krom seminearring together with the Kleene star operation is a Conway seminearring, henceforth called the \emph{Krom-Conway seminearring}. Explicitly, in addition to the axioms in \prettyref{t:seminearring} we have
\begin{align} 
	\label{eq:(K_cup_L)^ast}(K\cup L)^\ast &= (K^\ast L)^\ast K^\ast\\
	\label{eq:product-star} (KL)^\ast &= 1\cup K(LK)^\ast L.
\end{align}
\end{theorem}
\begin{proof} We first prove the sum-star-equation. Every word over the alphabet $\{K,L\}$ can be written uniquely in the form
\begin{align*}
	K^{i_1}L K^{i_2}L\cdots K^{i_m}L K^{i_{m+1}},
\end{align*}
where $m\geq 0$ and $i_1,\ldots,i_{m+1}\geq 0$. Hence
\begin{align*}
	(K\cup L)^\ast = \bigcup_{m\geq 0}\bigcup_{i_1,\ldots,i_{m+1}\geq 0} K^{i_1}L K^{i_2}L\cdots K^{i_m}L K^{i_{m+1}} = \bigcup_{m\geq 0}(K^\ast L)^m K^\ast = (K^\ast L)^\ast K^\ast.
\end{align*}

We now prove the product-star-equation. We compute
\begin{align*}
	(KL)^\ast = \bigcup_{n\geq 0}(KL)^n
	= 1\cup\bigcup_{n\geq 1}(KL)^n
	= 1\cup\bigcup_{n\geq 1}K(LK)^{n-1}L
	= 1\cup K\left(\bigcup_{m\geq 0}(LK)^m\right)L
	= 1\cup K(LK)^\ast L.
\end{align*}
\end{proof}

\begin{theorem} The starsemiring of proper Krom programs is a Conway semiring.
\end{theorem}
\begin{proof} Analogous to the proof of \prettyref{t:Conway}.
\end{proof}

\begin{corollary} For any Krom program $K$ and interpretation $I$,
\begin{align} 
	(K\cup I)^\ast = K^\ast\cup K^\ast I.
\end{align}
\end{corollary}
\begin{proof} We have
\begin{align*} 
	(K\cup I)^\ast \stackrel{\ref{eq:(K_cup_L)^ast}}= (K^\ast I)^\ast K^\ast = \left(\bigcup_{n\geq 0} (K^\ast I)^n\right)K^\ast.
\end{align*} Now since, for every $n\geq 1$,
\begin{align} 
	\label{eq:(K^astI)^n}(K^\ast I)^n = K^\ast IK^\ast I\ldots K^\ast I \stackrel{\ref{eq:IK=I}}= K^\ast I,
\end{align} we can simplify the above expression to
\begin{align*} 
	\left(\bigcup_{n\geq 0} (K^\ast I)^n\right)K^\ast \stackrel{\ref{eq:(K^astI)^n}}= (1\cup K^\ast I)K^\ast \stackrel{\ref{eq:(K_cup_L)M}}= K^\ast\cup K^\ast IK^\ast \stackrel{\ref{eq:IK=I}}= K^\ast\cup K^\ast I.
\end{align*}
\end{proof}

\section{Krom-Conway omegaseminearrings}\label{sec:omega}

We follow the algebraic tradition of semiring-based automata theory \cite{Esik06} and study a seminearring equipped with an additional unary omega operation.

In the rest of the paper, we will be concerned with the Krom algebra $$\mathfrak K_A := ( \mathbb K_A, \cup, \circ, \emptyset, 1_A, {}^\ast, {}^\omega )$$ of Krom programs under set union and composition together with the Kleene star and the least model ${}^\omega$-operation defined as
\begin{align*} 
	K^\omega := \bigcup_{n\geq 1} (K^n \emptyset),
\end{align*} henceforth called the \emph{Krom-Conway omegaseminearring} over $A$.



The next result provides an explicit characterization of the ${}^\omega$-operator. It shows that the least model of a Krom program is obtained by repeatedly propagating its facts through its proper part.

\begin{theorem} For any Krom program $K$, we have
\begin{align} 
	\label{eq:K^omega} K^\omega = p(K)^\ast f(K).
\end{align} 
\end{theorem}
\begin{proof} 
\begin{align*} 
	K^\omega
		& = \bigcup_{n\geq 1}(K^n\emptyset)\\
		& \stackrel{\ref{eq:f(K)=K0}}= \bigcup_{n\geq 1}f(K^n)\\
		&\stackrel{\ref{eq:f(K1-Kn)}}= \bigcup_{n\geq 1}\left(f(K)\cup \bigcup_{i=1}^{n-1}(p(K^i)f(K))\right)\\
		&\stackrel{\ref{eq:p(KL)}}= f(K)\cup\bigcup_{n\geq 1}\bigcup_{i=1}^{n-1}(p(K)^i f(K))\\
		&= f(K)\cup\bigcup_{n\geq 1}(p(K)^n f(K))\\
		&= \bigcup_{n\geq 0}(p(K)^n f(K))\\
		&= \left(\bigcup_{n\geq 0}p(K)^n\right)f(K)\\
		&= p(K)^\ast f(K).
\end{align*}
\end{proof}

\begin{proposition} For every Krom program $K$,
\begin{align} 
	\label{eq:f(K^ast)} K^\omega = f(K^\ast).
\end{align}
\end{proposition}
\begin{proof} We compute
\begin{align*} 
	f(K^\ast) 
		&= f(\bigcup_{n\geq 0} K^n)\\
		& \stackrel{\ref{eq:f(K_cup_L)}}= \bigcup_{n\geq 0} f(K^n)\\
		& \stackrel{\ref{eq:f(K1-Kn)}}= \bigcup_{n\geq 0}\left(f(K)\cup \bigcup_{i=1}^{n-1} p(K)^i f(K)\right)\\
		&= f(K)\cup \bigcup_{n\geq 0}\bigcup_{i=1}^{n-1} p(K)^i f(K)\\
		&= f(K)\cup \bigcup_{n\geq 1} p(K)^n f(K)\\
		& \stackrel{\ref{eq:(K_cup_L)M}}= f(K)\cup \left(\bigcup_{n\geq 1} p(K)^n\right)f(K)\\
		&= f(K)\cup p(K)^+ f(K)\\
		&= p(K)^\ast f(K)\\
		& \stackrel{\ref{eq:K^omega}}= K^\omega.
\end{align*}
\end{proof}

\begin{corollary} For any Krom program $K$,
\begin{align*} 
	K^\omega\subseteq K^\ast
\end{align*}
\end{corollary}
\begin{proof} A direct consequence of \prettyref{eq:f(K^ast)}.
\end{proof}

\begin{corollary} The ${}^\omega$-operator is compatible with Kleene star in the in the sense that
\begin{align*} 
	(K^\ast)^\omega = (K^\omega)^+
\end{align*} holds for every Krom program $K$.
\end{corollary}
\begin{proof} We have
\begin{align*} 
	(K^\ast)^\omega \stackrel{\ref{eq:f(K^ast)}}= f((K^\ast)^\ast) = f(K^\ast) = K^\omega = (K^\omega)^+,
\end{align*} where the second identity follows from the idempotency of the Kleene star and the last identity follows from the fact that Kleene plus is the identity on interpretations.
\end{proof}

We are now in the position to prove that the ${}^\omega$-operator captures the least-model semantics of Krom programs; cf.~\cite[Thm.~40]{Antic21-1}.

\begin{theorem}
For every Krom program $K$, we have
\[
	K^\omega = LM(K).
\]
\end{theorem}

\begin{proof}
Recall the van Emden-Kowalski immediate consequence operator~\cite{vanEmden76}
\[
	T_K(I) := \{a\in A\mid a\leftarrow B\in K,\ B\subseteq I\}.
\]
This operator is represented by sequential composition in the sense that
\[
	T_K(I) = KI
\]
for every interpretation $I$; see~\cite[Thm.~35]{Antic21-1}. Hence the bottom-up iteration of $T_K$ from the empty interpretation is given by
\[
	T_K^n(\emptyset) = K^n\emptyset
\]
for every $n\geq 0$. It is well-known that the least model of a logic program is obtained by a least fixed point iteration of its van Emden-Kowalski operator~\cite{vanEmden76}. Therefore,
\[
	LM(K) = \bigcup_{n\geq 1} T_K^n(\emptyset) = \bigcup_{n\geq 1} K^n\emptyset = K^\omega.
\]
\end{proof}



\subsection{Omega and composition}

A natural question is how the ${}^\omega$-operator behaves with respect to sequential composition. The following theorem provides an explicit formula for the least model of a composition $KL$, showing that it can be computed entirely from the proper and factual components of the factors.

\begin{corollary} For any Krom programs $K$ and $L$,
\begin{align}
	\label{eq:(KL)^omega} (KL)^\omega = p(KL)^\ast f(K)\cup p(KL)^\ast p(K)f(L).
\end{align}
\end{corollary}
\begin{proof} By the formula for the ${}^\omega$-operator in \prettyref{eq:K^omega} applied to $KL$, we have
\begin{align*} 
	(KL)^\omega = p(KL)^\ast f(KL) \stackrel{\ref{eq:p(KL)}, \ref{eq:f(KL)}}= p(KL)^\ast (f(K)\cup p(K)f(L)) \stackrel{\ref{eq:M(K_cup_L)}}= p(KL)^\ast f(K)\cup p(KL)^\ast p(K)f(L).
\end{align*}
\end{proof}

\subsection{Omega and set union}

The explicit characterization of the ${}^\omega$-operator obtained above immediately yields corresponding formulas for compound programs. We consider here the interaction of ${}^\omega$ with set union and derive an expression for the least model of the union of two Krom programs in terms of the least models of its parts.

\begin{corollary}\label{c:(K_cup_L)^omega} For any Krom programs $K$ and $L$,
\begin{align} 
	\label{eq:(K_cup_L)^omega} (K\cup L)^\omega = p(K^\ast L)^\ast K^\omega\cup p(L^\ast K)^\ast L^\omega.
\end{align}
\end{corollary}
\begin{proof} By the formula for the ${}^\omega$-operator in \prettyref{eq:K^omega} applied to $K\cup L$, we have
\begin{align*} 
	(K\cup L)^\omega
		&= p(K\cup L)^\ast f(K\cup L)\\
		&= (p(K)\cup p(L))^\ast (f(K)\cup f(L))\\ 
		& \stackrel{\ref{eq:M(K_cup_L)}}= (p(K)\cup p(L))^\ast f(K)\cup (p(K)\cup p(L))^\ast f(L)\\
		& \stackrel{\ref{eq:(K_cup_L)^ast}}= (p(K)^\ast p(L))^\ast p(K)^\ast f(K)\cup (p(L)^\ast p(K))^\ast p(L)^\ast f(L)\\
		& \stackrel{\ref{eq:K^omega}}= (p(K)^\ast p(L))^\ast K^\omega\cup (p(L)^\ast p(K))^\ast L^\omega\\
		& \stackrel{\ref{eq:p(K^ast)}}= (p(K^\ast)p(L))^\ast K^\omega\cup (p(L^\ast)p(K))^\ast L^\omega\\
		& \stackrel{\ref{eq:p(KL)}}= p(K^\ast L)^\ast K^\omega\cup p(L^\ast K)^\ast L^\omega.
\end{align*}
\end{proof}

\begin{corollary} For any Krom program $K$ and interpretation $I$,
\begin{align} 
	\label{eq:(K_cup_I)^omega} (K\cup I)^\omega = K^\omega\cup p(K)^\ast I.
\end{align}
\end{corollary}
\begin{proof}
\begin{align*}
	(K\cup I)^\omega
		& \stackrel{\ref{eq:K^omega}}= p(K\cup I)^\ast f(K\cup I)\\
		&\stackrel{\ref{eq:p(K_cup_L)},\ref{eq:f(K_cup_L)}}= (p(K)\cup p(I))^\ast(f(K)\cup f(I))\\
		&= p(K)^\ast(f(K)\cup I)\\
		&= p(K)^\ast f(K)\cup p(K)^\ast I\\
		&\stackrel{\ref{eq:K^omega}}= K^\omega\cup p(K)^\ast I.
\end{align*}
\end{proof}




\section{Permutation programs}\label{sec:permutations}

This brief section is concerned with permutations and their representation as Krom programs:

\begin{definition} We associate with every permutation $\pi : A\to A$ the \emph{permutation program}
\begin{align*} 
	K_\pi := \{\pi(a) \leftarrow a \mid a\in A\}.
\end{align*} We denote the set of all permutation programs over $A$ by $\Pi_A$.
\end{definition}


\begin{example} Let $A := \{a, b, c, d\}$. Then the Krom program
\begin{align*} 
	K_{(ab)(cd)} = \left\{
	\begin{array}{c}
		a \leftarrow b\\
		b \leftarrow a\\
		c \leftarrow d\\
		d \leftarrow c
	\end{array}
	\right\}
\end{align*} is the permutation program corresponding to the permutation $(ab)(cd)$ written in cycle notation.
\end{example}


\begin{theorem} For any permutations $\pi$ and $\theta$ of $A$, we have
\begin{align*} 
	K_\pi K_\theta &= K_{\pi\theta}\\
	K_{id_A} &= 1_A\\
	K_{\pi^{-1}} &= K_\pi^d,
\end{align*} where the dual $K_\pi^d$ is obtained from $K_\pi$ by reverting all arrows in $K_\pi$.\footnote{This is not a problem since $K_\pi$ consists only of proper rules.} That is, the mapping $\pi\mapsto K_\pi$ is a group isomorphism from the permutation group $S_A$ over $A$ to the group of all permutation programs over $A$ thus showing
\begin{align*} 
	(S_A, \circ, id_A, {}^{-1})\cong (\Pi_A, \circ, 1_A, {}^d).
\end{align*}
\end{theorem}


\begin{proposition} Two distinct permutation programs are never subsumption equivalent.
\end{proposition}
\begin{proof} 

Let $K_\pi$ and $K_\sigma$ be two distinct permutation programs over an alphabet $A$. Then there exists an atom $a\in A$ such that $\pi(a)\neq\sigma(a)$. Hence
\[
K_\pi\circ\{a\}=\{\pi(a)\}
\]
and
\[
K_\sigma\circ\{a\}=\{\sigma(a)\}.
\]
Since $\pi(a)\neq\sigma(a)$, it follows that
\[
K_\pi\circ\{a\}\neq K_\sigma\circ\{a\},
\]
which implies
\[
K_\pi\not\equiv_{ss} K_\sigma.
\]
\end{proof}

\section{Generators and canonical decompositions}\label{sec:generators}

In this section, we investigate generating sets for finite Krom seminearrings from two complementary perspectives. First, by fixing an enumeration of the underlying finite alphabet, we construct a natural generating set that yields shortest decompositions of singleton Krom programs and hence canonical decompositions of arbitrary Krom programs as unions of such singleton decompositions. Second, we show that, despite the size of this canonical generating set, every finite Krom seminearring can in fact be generated by only three elements. Thus, while canonical decompositions require a richer collection of generators, remarkably compact generating sets also exist.

\begin{theorem}\label{t:bbG_A}
Fix an enumeration $A=\{a_0,\ldots,a_n\}$, $n\geq 0$, of a finite alphabet $A$. Let $$\mathfrak K_A:=(\mathbb K_A,\cup,\circ,\emptyset,1_A)$$ be the corresponding finite Krom seminearring. The set
\begin{align}\label{eq:bbG_A}
	\mathbb G_A:=\{\{a_0\}\}\cup\{\{a_{i+1}\leftarrow a_i\}\mid 0\leq i\leq n-1\}\cup\{\{a_0\leftarrow a_n\}\}
\end{align}
generates $\mathbb K_A$ as a seminearring. Moreover, every Krom program is a union of shortest decompositions of singleton Krom programs over $\mathbb G_A$.
\end{theorem}
\begin{proof}
Consider the directed cycle $$a_0\to a_1\to\cdots\to a_n\to a_0$$ induced by the proper-rule generators in $\mathbb G_A$. For $0\leq i,j\leq n$, let $d(i,j)$ denote the length of the unique directed path from $a_i$ to $a_j$ in this cycle.

First, every singleton fact $\{a_i\}$ belongs to the subseminearring generated by $\mathbb G_A$. Indeed, for $i=0$ this is immediate since $\{a_0\}\in\mathbb G_A$, and for $1\leq i\leq n$ we have
\begin{align*}
	\{a_i\} &= \{a_i\leftarrow a_{i-1}\}\circ\cdots\circ\{a_1\leftarrow a_0\}\circ\{a_0\}.
\end{align*}
This decomposition uses exactly $d(0,i)$ proper-rule generators.

Next, let $\{a_j\leftarrow a_i\}$ be a proper singleton Krom program. If $i<j$, then
\begin{align*}
	\{a_j\leftarrow a_i\} &= \{a_j\leftarrow a_{j-1}\}\circ\cdots\circ\{a_{i+1}\leftarrow a_i\},
\end{align*}
which uses exactly $d(i,j)=j-i$ generators. If $j<i$, then the unique directed path from $a_i$ to $a_j$ passes through $a_n$ and wraps around to $a_0$, yielding
\begin{align*}
	\{a_j\leftarrow a_i\} &= \{a_j\leftarrow a_{j-1}\}\circ\cdots\circ\{a_1\leftarrow a_0\}\circ\{a_0\leftarrow a_n\}\circ\{a_n\leftarrow a_{n-1}\}\circ\cdots\circ\{a_{i+1}\leftarrow a_i\},
\end{align*}
which uses exactly $d(i,j)=n-i+j+1$ generators.

Hence every singleton Krom program can be expressed in terms of $\mathbb G_A$. Since every Krom program is the union of its singleton subprograms, it follows that $\mathbb G_A$ generates $\mathbb K_A$ as a seminearring.

It remains to prove that the above decompositions are shortest. Every proper-rule generator corresponds to exactly one edge of the directed cycle $$a_0\to a_1\to\cdots\to a_n\to a_0.$$ Thus every composition of $m$ proper-rule generators corresponds to a directed path of length $m$. Consequently, no singleton rule $\{a_j\leftarrow a_i\}$ can be obtained using fewer than $d(i,j)$ proper-rule generators. Similarly, no singleton fact $\{a_i\}$ can be obtained using fewer than $d(0,i)$ proper-rule generators together with the unique fact generator $\{a_0\}$. Hence all singleton decompositions displayed above are shortest.

Finally, let $K\in\mathbb K_A$ be arbitrary. Since composition distributes over set union, every decomposition of $K$ over $\mathbb G_A$ can be rewritten as a union of compositions representing singleton Krom programs. By the arguments above, each such singleton decomposition must have length at least equal to the corresponding directed distance. Hence replacing every singleton subprogram of $K$ by its shortest decomposition yields a shortest decomposition of $K$. Therefore every Krom program is a union of shortest decompositions of its singleton Krom subprograms over $\mathbb G_A$.
\end{proof}

\begin{proposition} 
The generating set $\mathbb G_A$ in \prettyref{eq:bbG_A} is minimal.
\end{proposition}
\begin{proof} We show that no element of $\mathbb G_A$ can be generated by the remaining ones.

First, the generator $\{a_0\}$ cannot be omitted, since all other generators are proper programs and the composition and union of proper programs is again proper. Hence no fact can be generated without $\{a_0\}$.

Next, fix $0\leq i\leq n-1$. If the generator $\{a_{i+1}\leftarrow a_i\}$ is removed, then every remaining proper-rule generator corresponds to an edge of the directed cycle $$a_0\to a_1\to\cdots\to a_n\to a_0$$ except for the edge from $a_i$ to $a_{i+1}$. Since composition of singleton proper rules corresponds to concatenation of directed paths, every proper singleton rule generated by the remaining proper-rule generators corresponds to a directed path in the cycle with this edge removed. But there is no directed path from $a_i$ to $a_{i+1}$ in this graph. Therefore $\{a_{i+1}\leftarrow a_i\}$ cannot be generated by the remaining generators.

The same argument applies to the generator $\{a_0\leftarrow a_n\}$: if it is removed, then the edge from $a_n$ to $a_0$ is missing, and there is no directed path from $a_n$ to $a_0$ using the remaining proper-rule generators. Hence $\{a_0\leftarrow a_n\}$ cannot be generated by the remaining generators.

Thus no generator in $\mathbb G_A$ lies in the subseminearring generated by the other generators. Therefore $\mathbb G_A$ is minimal.
\end{proof}


\begin{proposition}
For every alphabet $A$ with $|A|\geq 2$, the finite Krom seminearring $\mathfrak K_A$ is generated by three elements.
\end{proposition}
\begin{proof}
Fix an enumeration $A=\{a_0,\ldots,a_n\}$ with $n\geq 1$, and define $$C:=\{a_{i+1}\leftarrow a_i\mid 0\leq i<n\}\cup\{a_0\leftarrow a_n\},$$ $$E:=\{a_0\leftarrow a_0\},$$ and $$F:=\{a_0\}.$$ We claim that $C,E,F$ generate $\mathbb K_A$.

Using indices modulo $n+1$, the program $C$ represents the cyclic permutation $a_i\mapsto a_{i+1}$. Hence,
\begin{align*} 
	C^r=\{a_{i+r}\leftarrow a_i\mid 0\leq i\leq n\}
\end{align*} for every $r\geq 0$. Therefore, for all $0\leq i,j\leq n$, we have
\begin{align*}
	C^j\circ E\circ C^{n+1-i} &= \{a_j\leftarrow a_i\}.
\end{align*}
Thus every proper singleton Krom program is generated by $C$ and $E$. Moreover, for every $0\leq j\leq n$, we have
\begin{align*}
	C^j\circ E\circ F &= \{a_j\}.
\end{align*}
Hence every singleton fact is generated by $C,E,F$. Since every Krom program is a finite union of singleton Krom programs, it follows that $C,E,F$ generate $\mathbb K_A$.
\end{proof}

\section{Automata}\label{sec:automata}

In this section, we show that finite Krom monoids have enough algebraic structure to represent finite-state machines. More precisely, we establish that every deterministic semiautomaton naturally gives rise to a Krom monoid isomorphic to its transformation monoid.

Recall that a \emph{(deterministic) semiautomaton} is a triple $\mathfrak A=(Q,\Sigma,\delta)$, where $Q$ is a finite set of states, $\Sigma$ is a finite alphabet, and $\delta:Q\times\Sigma\to Q$ is the transition function.

For every $a\in\Sigma$, define the Krom program $$K_a:=\{\delta(q,a)\leftarrow q\mid q\in Q\},$$ and extend this assignment recursively to words by
\begin{align*}
	K_\varepsilon &:= 1_Q,\\
	K_{uv} &:= K_u\circ K_v,
\end{align*}
where $\varepsilon$ denotes the empty word and $u,v\in\Sigma^\ast$.

Similarly, define
\begin{align*}
	\delta_\varepsilon &:= id_Q,\\
	\delta_{uv} &:= \delta_u\circ\delta_v,
\end{align*}
and let $$\mathfrak T_ \mathfrak A:=(\{\delta_w\mid w\in\Sigma^\ast\},\circ,\delta_\varepsilon)$$ denote the transformation monoid of $\mathfrak A$. Furthermore, let $$\mathbb K_ \mathfrak A:=\{K_w\mid w\in\Sigma^\ast\}$$ and define the associated Krom monoid by $$\mathfrak K_ \mathfrak A:=(\mathbb K_ \mathfrak A,\circ,1_Q).$$

We then have the following correspondence between the Krom monoid and the transformation monoid of a semiautomaton:

\begin{theorem}\label{t:transformation-krom}
The mapping $$K_w\mapsto\delta_w,$$ for $w\in\Sigma^\ast$ is a monoid isomorphism from $\mathfrak K_ \mathfrak A$ onto the transformation monoid $\mathfrak T_ \mathfrak A$.
\end{theorem}

We now turn our attention to deterministic finite automata. Recall that such an automaton is a tuple $\mathfrak A=(Q,\Sigma,\delta,q_0,F)$, where $(Q,\Sigma,\delta)$ is a semiautomaton, $q_0\in Q$ is the initial state, and $F\subseteq Q$ is the set of final states. Its behavior is given by $$||\mathfrak A||:=\{w\in\Sigma^\ast\mid\delta(q_0,w)\in F\}.$$

The above theorem immediately yields the following characterization:

\begin{proposition}
For every deterministic finite automaton $\mathfrak A=(Q,\Sigma,\delta,q_0,F)$, $$||\mathfrak A||=\{w\in\Sigma^\ast\mid F\cap(K_w\circ{q_0})\neq\emptyset\}.$$
\end{proposition}
\begin{proof}
By construction, $$K_w=\{\delta(q,w)\leftarrow q\mid q\in Q\}.$$ Since $\mathfrak A$ is deterministic, for every state $q\in Q$ the program $K_w$ contains exactly one rule whose body is $q$, namely $\delta(q,w)\leftarrow q$. In particular, $$K_w\circ{q_0}=\{\delta(q_0,w)\}.$$ Consequently, $$F\cap(K_w\circ{q_0})\neq\emptyset$$ if and only if $$\delta(q_0,w)\in F,$$ which is equivalent to $$w\in||\mathfrak A||.$$
\end{proof}

\begin{corollary}
Every finite monoid is isomorphic to a submonoid of a finite Krom monoid.
\end{corollary}
\begin{proof}
By the Cayley theorem for monoids, every finite monoid is isomorphic to a submonoid of a finite transformation monoid. Furthermore, every finite transformation monoid is the transformation monoid of a finite semiautomaton and is therefore, by \prettyref{t:transformation-krom}, isomorphic to a finite Krom monoid. The result now follows by composition of the two embeddings.
\end{proof}


\section{Conclusion}

In this paper, we investigated the algebraic structure induced by sequential composition on propositional Krom logic programs. We showed that, despite their syntactic simplicity, Krom programs give rise to a remarkably rich collection of algebraic structures, including monoids, seminearrings, quemirings, and Conway seminearrings and omegaseminearrings. Along the way, we established explicit generating sets and canonical decompositions, analyzed the interaction of sequential composition with the facts, proper-rule, Kleene star, and ${}^\omega$-operators, and related Krom programs to directed graphs, transformation monoids, and finite-state automata.

These results provide further evidence that logic programs can be studied fruitfully from an algebraic perspective. In particular, the occurrence of Conway identities and graph-theoretic interpretations suggests close connections with semiring-based automata theory and algebraic graph theory. Moreover, the fact that Krom programs arise as the basic components in the Minimalist Decomposition Theorem for Horn programs~\cite{Antic24-1} provides further motivation for investigating their algebraic structure.

The present work raises several directions for future research. A natural next step is to investigate homomorphisms, congruences, and quotient constructions for logic program algebras and to develop a corresponding representation theory. It would also be interesting to extend the algebraic techniques developed here beyond the Krom fragment and to study analogous structures for richer classes of logic programs and answer set programs.

\bibliographystyle{acm}
\bibliography{/Users/christianantic/Bibdesk/Bibliography,/Users/christianantic/Bibdesk/Publications_J,/Users/christianantic/Bibdesk/Publications_C,/Users/christianantic/Bibdesk/Preprints,/Users/christianantic/Bibdesk/Submitted,/Users/christianantic/Bibdesk/Papers,/Users/christianantic/Bibdesk/Bin}
\if\isdraft1\newpage

\section{Strong equivalence}

\begin{definition} Two Krom programs $K$ and $L$ are \emph{strongly equivalent} \cite{Lifschitz01} --- in symbols, $K\equiv_s L$ --- if $K\cup M\equiv_\omega L\cup M$ holds for every Krom program $M$.
\end{definition}

We have
\begin{align*} 
	(K\cup M)^\omega = (K\cup M)^\ast\emptyset = (K^\ast M)^\ast K^\ast\emptyset = (K^\ast M)^\ast K^\omega,
\end{align*} which gives raise to the following algebraic characterization of strong equivalence in terms of sequential composition and the Kleene star and omega operations:

\begin{theorem}\label{t:G_equiv_s_H} Given two graphs (i.e., proper Krom programs) $G$ and $H$,
\begin{align*} 
	G\equiv_s H \quad\Leftrightarrow\quad (G^\ast M)^\ast\equiv_\omega (H^\ast M)^\ast,\quad\text{for every Krom program $M$}.
\end{align*}
\end{theorem}
\begin{proof} We have
\begin{align*} 
	G\equiv_s H 
		\quad& \stackrel{\ref{eq:}}\Leftrightarrow\quad (G^\ast M)^\ast G^\omega = (H^\ast M)^\ast H^\omega\\
		\quad& \stackrel{\ref{eq:}}\Leftrightarrow\quad (G^\ast M)^\ast \emptyset = (H^\ast M)^\ast \emptyset\\
		\quad& \stackrel{\ref{eq:}}\Leftrightarrow\quad (G^\ast M)^\ast \equiv_\omega (H^\ast M)^\ast
\end{align*} for every Krom program $M$.
\end{proof}

\section{Commutatators}

\begin{definition} We say that $K$ and $L$ \emph{commute} if $KL=LK$.
\end{definition}

The Krom programs $K$ and $L$ commute iff
\begin{align*} 
	f(KL) = f(LK) \quad\text{and}\quad p(KL) = p(LK),
\end{align*} which is equivalent to
\begin{align*} 
	f(K)\cup p(K)f(L) = f(L)\cup p(L)f(K) \quad\text{and}\quad p(K)p(L) = p(L)p(K).
\end{align*}

\begin{lemma} If $K$ commutes with $L$, then $p(K)$ commutes with $p(L)$.
\todo[inline]{}
\end{lemma}
\begin{proof} 
\todo[inline]{}
\end{proof}

\begin{proposition} For any commuting\todo{es reicht ss-commutation} Krom programs $K$ and $L$ we have
\begin{align*} 
	(K\cup L)^\omega = ... 
\end{align*}
\todo[inline]{}
\end{proposition}
\begin{proof} 
\todo[inline]{}
\end{proof}

\subsection{}

\begin{definition} We say that $K$ and $L$ \emph{$\star$-commute}, for $\star\in \{\omega, ss, u, s\}$, if $KL\equiv_\star LK$.
\end{definition}

\section{Idempotent programs}\label{sec:Idempotents}

In finite monoids, idempotent elements $e$ such that $e^2 = e$ play a structurally central role. They are not peripheral artifacts; they organize the algebraic, combinatorial, and representation-theoretic structure of the monoid and in this section, we thus study idempotents in finite Krom monoids, where we denote the set of all idempotent Krom programs by $\mathbb E$\todo{needed?}.


Our first observation is that computing the omega of an idempotent Krom program is trivial in the following sense:

\begin{proposition} For every idempotent Krom program $K$,
\begin{align*} 
	K^\omega = K \emptyset.
\end{align*}
\end{proposition}
\begin{proof} 
\begin{align*} 
	K^\omega = \bigcup_{n\geq 0}f(K^n) = \bigcup_{n\geq 0} f(K) = f(K) \stackrel{\ref{eq:f(K)=K0}}= K\circ \emptyset.
\end{align*}
\end{proof}

\begin{proposition} Every interpretation is idempotent.
\end{proposition}
\begin{proof} A direct consequence of the fact that interpretations are left zeros \prettyref{eq:IK=I}.
\end{proof}

\begin{proposition} A Krom program $K$ is idempotent iff
\begin{align*} 
	p(K)f(K)\subseteq f(K) \quad\text{and}\quad p(K) = p(K)p(K).
\end{align*}
\end{proposition}
\begin{proof} We have
\begin{align*} 
	f(K) = f(K^2)\stackrel{\ref{eq:f(KL)}} = f(K)\cup p(K)f(K)
\end{align*} ...
\todo[inline]{}
\end{proof}

\begin{example} 
\todo[inline]{}
\end{example}

\subsection{}

\begin{definition} A Krom program $K$ is \emph{$\star$-idempotent}, for $\star\in\{\omega, s, u, ss\}$, if $K\equiv_\star K^2$.
\end{definition}

\begin{proposition} Every Krom program is $\omega$-idempotent.
\end{proposition}
\begin{proof} 
\todo[inline]{}
\end{proof}

A Krom program $K$ is $ss$-idempotent iff
\begin{align*} 
	KI = K(KI)
\end{align*} holds for every interpretation $I$...

\section{Quasi-invertible programs}

\begin{definition} We call $L$ a \emph{quasi-inverse} of $K$ if $K = KLK$, in which case we call $K$ \emph{quasi-invertible}.\footnote{In the semigroup literature, this property is usually called ``regular'' (cf. \cite{Howie03}).}
\end{definition}

\begin{fact} In case $L$ is a quasi-inverse of $K$, the product $KL$ is idempotent, that is,
\begin{align*} 
	(KL)^2 = (KL)(KL) = (KLK)L = KL.
\end{align*}
\end{fact}

We can characterize quasi-invertible programs as follows: 

\begin{proposition} By \prettyref{eq:p(KL)} and \prettyref{eq:f(KL)}, we have
\begin{align*} 
	K = KLK \quad\Leftrightarrow\quad
		f(K) &= f(KLK) = f(K)\cup p(K)f(L)\cup p(K)p(L)f(K),\\
		p(K) &= p(KLK) = p(K)p(L)p(K),
\end{align*} where the first line is equivalent to
\begin{align*} 
	p(K)f(L)\subseteq f(K) \quad\text{and}\quad p(K)p(L)f(K)\subseteq f(K).
\end{align*}
\end{proposition}

\begin{proposition} In case $L$ is a quasi-inverse of $K$, we have $$(KL)^2=KL,$$ which means that $KL$ is idempotent and thus
\begin{align*} 
	(KL)^\omega \stackrel{\ref{eq:K^omega=f(K)}}= f(KL) \stackrel{\ref{eq:f(KL)}}= f(K)\cup p(K)f(L).
\end{align*}
\end{proposition}

\begin{proposition} For any Krom programs $K$ and $L$, if $L$ is a quasi-inverse\todo{es reicht ss-quasi-inverse} of $K$, then
\begin{align*} 
	(K\cup L)^\omega = ...
\end{align*}
\todo[inline]{}
\end{proposition}
\begin{proof} 
\todo[inline]{}
\end{proof}

\section{Invertible programs}

\begin{definition} We call a Krom program $K$ \emph{invertible} if there is a program $L$ such that $KL = 1_{h(p(K))}$.
\todo[inline]{was ist mit $LK = 1_{\ldots}$?}
\end{definition}

\begin{fact} Every permutation program is invertible.
\end{fact}

\section{Elevators}\label{sec:elevators}

\begin{definition} 
\begin{align*} 
	E_{(a_0, \ldots, a_n)} := \{a_0\}\cup \{a_{i+1} \leftarrow a_i \mid 1\leq i\leq n - 1\}
\end{align*} $\mathbb{ELEVATORS}$
\todo[inline]{}
\end{definition}

\todo[inline]{Welche Programme lassen sich mit Elevators generieren? Dh wie sieht $\langle \mathbb E\rangle$ aus?}

\section{Congruences}

\begin{definition} Recall that a \emph{congruence} ... 
\todo[inline]{}
\end{definition}

\begin{theorem} Subsumption equivalence is a congruence with respect to sequential composition.
\end{theorem}
\begin{proof} 
\todo[inline]{}
\end{proof}

\section{Aperiodic programs}\label{sec:aperiodic}

\begin{definition} A Krom program $K$ is \emph{aperiodic} if $K^{n+1} = K^n$, for some $n\geq 1$.
\end{definition}

\begin{example} 
\todo[inline]{}
\end{example}

\begin{definition} We define the \emph{closure} of a Krom program $K$ by
\begin{align*} 
	\bar K := 1\cup K.
\end{align*}
\end{definition}

Notice that the closure of $K$ is $\omega$-equivalent to $K$, that is,
\begin{align*} 
	\bar K\equiv_\omega K.
\end{align*}

\begin{proposition} The closure of every Krom program is aperiodic.
\end{proposition}
\begin{proof} We have
\begin{align*} 
	(1\cup K)^{m+1} = (1\cup K)(1\cup K)^m \stackrel{\ref{eq:(K_cup_L)M}}= (1\cup K)^m\cup K(1\cup K)^m,
\end{align*} which shows
\begin{align*} 
	(1\cup K)^m\subseteq (1\cup K)^{m+1},
\end{align*} equivalent to
\begin{align*} 
	\bar K^m\subseteq \bar K^{m+1},
\end{align*} for all $m\geq 1$. Since the underlying alphabet is finite, there must be some $n\geq 1$ such that $\bar K^n = \bar K^{n+1}$.
\end{proof}

\begin{fact} Every elevator program is aperiodic.
\end{fact}

\begin{fact} The only aperiodic permutation program is the unit program.
\end{fact}

\section{Nilpotent programs}\label{sec:nilpotent}

\begin{definition} A Krom program $K$ is \emph{nilpotent} if $K^n = \emptyset$, for some $n\geq 1$.
\end{definition}

\begin{proposition} Every nilpotent Krom program is proper.
\end{proposition}
\begin{proof} A direct consequence of the fact that $f(K)\subseteq K^n$, for all $n\geq 1$.
\end{proof}

\begin{example} Every single-rule Krom program $\{a \leftarrow b\}$, $a\neq b$, is nilpotent since
\begin{align*} 
	\{a \leftarrow b\}\circ \{a \leftarrow b\} = \emptyset.
\end{align*}
\end{example}

\section{Transformation semigroups}

\todo[inline]{Zeige, dass jede t.s. durch eine Krom-HG dargestellt werden kann --- ist das durch \prettyref{sec:automata} subsumiert?}

\section{Index and period}

In this section, we compute the index and period of a Krom program as motivated by the following characterization of the least model semantics:

\begin{theorem} For every Krom program $K$,
\begin{align*} 
	K^\omega = K^{i(K)} \emptyset.
\end{align*}
\end{theorem}
\begin{proof} A direct consequence of \prettyref{eq:f(K)=K0} and the definition of $K^\omega$ in \prettyref{eq:K^omega}.
\end{proof}

\begin{proposition} For any Krom program $K$ and homomorphism ${}^\bullet$, we have
\begin{align*} 
	i(K)\geq i(K^\bullet).
\end{align*}
\end{proposition}
\begin{proof}
\todo[inline]{}
\end{proof}

\section{Homomorphisms}

\todo[inline]{Hier oder fuer eigenes Paper \cite{Antic26-2} aufheben?}

\section{Least model equivalence}

\begin{definition} Two Krom programs $K$ and $L$ are \emph{${}^\omega$-equivalent} if $K^\omega = L^\omega$.
\end{definition}

In order to understand ${}^\omega$-equivalence we first need to understand the computation of $K^\omega$. [todo das wurde oben doch schon gemacht!] We have
\begin{align*} 
	K^\omega = f(K^+) = K^+\emptyset = \left(\bigcup_{n\geq 1}K^n\right)\emptyset = \bigcup_{n\geq 1}(K^n\emptyset) = \bigcup_{n\geq 1}f(K^n).
\end{align*} Let us compute $K^n$. We have
\begin{align*} 
	K^2
		&= (f(K)\cup p(K))(f(K)\cup p(K))\\
		&= f(K)^2\cup p(K)f(K)\cup f(K)p(K)\cup p(K)^2\\
		& \stackrel{\ref{eq:IK=I}}= f(K)\cup p(K)f(K)\cup p(K),
\end{align*} and
\begin{align*} 
	K^3 = K^2K = f(K)\cup p(K)f(K)\cup p(K)^2K = f(K)\cup p(K)f(K)\cup p(K)^2f(K)\cup p(K)^3,
\end{align*} and
\begin{align*} 
	K^4
		&= K^3K\\
		&= f(K)\cup p(K)f(K)\cup p(K)^2f(K)\cup p(K)^3K\\
		&= f(K)\cup p(K)f(K)\cup p(K)^2f(K)\cup p(K)^3f(K)\cup p(K)^4.
\end{align*} A simple proof by induction shows the general formula for arbitrary $n\geq 1$ given by
\begin{align}\label{eq:K^n} 
	K^n = f(K)\cup p(K)^n\cup \bigcup_{i=1}^{n-1} p(K)^i f(K).
\end{align} This implies
\begin{align} 
	K^\omega
		&= f(K^+)\\
		&= f(\bigcup_{n\geq 1}K^n)\\
		&= \bigcup_{n\geq 1}f(f(K)\cup p(K)^n\cup \bigcup_{i=1}^{n-1} p(K)^i f(K))\\
		&= \bigcup_{n\geq 1}\left(f(f(K))\cup f(p(K)^n)\cup\bigcup_{i=1}^{n-1}f(p(K)^i f(K))\right)\\
		&= f(K)\cup \bigcup_{n\geq 1}\bigcup_{i=1}^{n-1}p(K)^i f(K)\\
		\label{eq:K^omega_} &= \bigcup_{n\geq 1}\bigcup_{i=0}^{n-1} p(K)^i f(K)\\
		&= \bigcup_{n\geq 1} p(K)^i f(K)\\
		&= \bigcup_{i = 1}^{i(p(K))} p(K)^i f(K),
\end{align} where the fifth identity follows from $f(f(K))=f(K)$ and $f(p(K)^n)=\emptyset$, the sixth identity follows from $p(K)^0=1$, and the seventh identity holds since...

Hence, we obtain the following algebraic characterization of ${}^\omega$-equivalence: 

\begin{theorem} For any Krom programs $K$ and $L$,
\begin{align*} 
	K\equiv_\omega L \quad\Leftrightarrow\quad \bigcup_{i = 1}^{i(p(K))} p(K)^i f(K) = \bigcup_{i = 1}^{i(p(L))} p(L)^i f(L).
\end{align*}
\end{theorem}

\todo[inline]{wie kann man vorheriges thm anwenden?}

\section{Subsumption equivalence}\label{sec:ss}

\begin{definition}[\cite{Maher88}] Two Krom programs $K$ and $L$ are \emph{subsumption equivalent} --- in symbols, $K\equiv_{ss} L$ --- if $KI = LI$ holds for every interpretation $I$.\footnote{Subsumption equivalence was originally defined in terms of the van Emden Kowalski operator \cite{vanEmden76} $T_P$ of a logic program $P$, for which it is easy to show that we have $T_P(I)=PI$.}
\end{definition}

Interestingly, for Krom programs, subsumption equivalence turns out to be trivial:

\begin{theorem} For any Krom programs $K$ and $L$,
\begin{align} 
	\label{eq:K_equiv_ss_L} K\equiv_{ss} L \quad\Leftrightarrow\quad K = L.
\end{align}
\end{theorem}
\begin{proof} Our first observation is that in case $K$ and $L$ are subsumption equivalent, we must have
\begin{align*} 
	K\emptyset = L\emptyset
\end{align*} which by \prettyref{eq:f(K)=K0} is equivalent to
\begin{align*} 
	f(K) = f(L).
\end{align*} 

Now let $a$ be an atom in the head of $K$ which is not a fact of $K$, that is,
\begin{align*} 
	a\in h(K) \quad\text{and}\quad a\not\in f(K).
\end{align*} Then there must be a rule $a\leftarrow b\in K$, for some $b\in A$. Let $I:=\{b\}$. Since $KI = LI$ and $a\in KI$ imply $a\in LI$, we must have $a\leftarrow b\in L$. This immediately implies $K\subseteq L$. 

An analogous argument shows $L\subseteq K$.\todo{check}
\end{proof}

\todo[inline]{Wann gilt $KL\equiv_{ss} LK$?}

\section{Uniform equivalence}\label{sec:uniform_equivalence}

In this section, we are concerned with a form of equivalence often studied in logic programming and database theory:

\begin{definition}[\cite{Maher88,Sagiv88}] Two Krom programs $K$ and $L$ are \emph{uniformly equivalent} --- in symbols, $K\equiv_u L$ --- if $K\cup I\equiv_\omega L\cup I$ holds for every interpretation $I$.
\end{definition}

Uniform equivalence can be rephrased in terms of the omega operation as
\begin{align*} 
	K\equiv_u L \quad\Leftrightarrow\quad (K\cup I)^\omega = (L\cup I)^\omega,
\end{align*} for every interpretation $I$. This means that in order to understand uniform equivalence, we need to understand $(K\cup I)^\omega$.

We first study iterations of $K\cup I$:

\begin{lemma} For every Krom program $K$ and interpretation $I$,
\begin{align} 
	(K\cup I)^n = K^n\cup K^{n-1}I\cup\ldots \cup KI\cup I
\end{align} holds for every $n\geq 1$.
\end{lemma}
\begin{proof} By induction on $n\geq 1$. The induction base $n=1$ holds trivially. For the induction step, ... 
\end{proof}

We have the following characterization of uniform equivalence:

\begin{corollary} For any Krom programs $K$ and $L$,
\begin{align*} 
	K\equiv_u L \quad\Leftrightarrow\quad K^\omega\cup K^\ast I = L^\omega\cup L^\ast I,
\end{align*} for every interpretation $I$.
\end{corollary}
 
This means that, in order to understand uniform equivalence, we need to understand the ${}^\omega$- and ${}^\ast$-operators. 

We first wish to understand $K^\ast I$:

\begin{lemma} For any Krom program $K$ and interpretation $I$, we have
\begin{align} 
	\label{eq:K^astI} K^\ast I = K^\omega\cup p(K)^\ast I.
\end{align}
\end{lemma}

\begin{theorem} 
\begin{align*} 
	(K\cup I)^\omega \stackrel{\ref{eq:}}= K^\omega\cup K^\ast I \stackrel{\ref{eq:K^astI}}= K^\omega\cup p(K)^\ast I.
\end{align*} 
\end{theorem} 
\begin{proof} We compute
\begin{align*} 
	K^\ast I
		&=\left(\bigcup_{n\geq 0}K^n\right)I\\
		&\stackrel{\ref{eq:(K_cup_L)M}}= \bigcup_{n\geq 0}(K^nI)\\
		&\stackrel{\ref{eq:K^n}}= \bigcup_{n\geq 0}\left(f(K)I\cup p(K)^nI\cup\bigcup_{i=1}^{n-1}p(K)^if(K)I\right)\\
		&\stackrel{\ref{eq:IK=I}}= f(K)\cup\bigcup_{n\geq 0}\left(p(K)^nI\cup\bigcup_{i=1}^{n-1}p(K)^if(K)\right)\\
		&= f(K)\cup\bigcup_{n\geq 0}p(K)^nI\cup\bigcup_{n\geq 0}\bigcup_{i=1}^{n-1}p(K)^if(K)\\
		&\stackrel{\ref{eq:K^omega_}}= K^\omega\cup\bigcup_{n\geq 0}p(K)^nI\\ 
		&= K^\omega\cup p(K)^\ast I.
\end{align*} 
\end{proof}

Hence, we have finally arrived at the following algebraic characterization of uniform equivalence of propositional Krom logic programs:

\begin{theorem} For any Krom programs $K$ and $L$,
\begin{align} 
	\label{eq:K_equiv_u_L} K\equiv_u L \quad\Leftrightarrow\quad K^\omega\cup p(K)^\ast I = L^\omega\cup p(L)^\ast I,
\end{align} for every interpretation $I$.
\end{theorem}

\begin{corollary} For any graph programs $G$ and $H$,
\begin{align*} 
	G\equiv_u H \quad\Leftrightarrow\quad G^\ast I = H^\ast I,
\end{align*} for every interpretation $I$.
\end{corollary}
\begin{proof} A direct consequence of \prettyref{eq:K_equiv_u_L}, \prettyref{eq:G^omega=0}, and \prettyref{eq:p(G)=G}.
\end{proof}

Hat Corollary eine Bedeutung im graphtheoretischen Sinne?

\section{Directed graphs}\label{sec:graphs}

In the preliminaries, I have mentioned that proper Krom programs can be identified with graphs and in this section we analyze that connection in more detail. Throughout this section, we treat atoms and vertices as synonyms. 

In the remainder of this section, $G$ and $H$ always denote graphs represented as proper Krom programs.

\begin{fact}\label{f:reachability} A vertex $a$ is reachable from a vertex $b$ in $G$ iff $a \leftarrow b\in G^+$.
\end{fact}

\subsection{\texorpdfstring{${}^\omega$-Equivalence}{}}

...

\subsection{Subsumption equivalence}

...

\subsection{Uniform equivalence}

Since graph programs are proper and thus contain no facts, the least model of a graph program $G$ is always empty,
\begin{align*} 
	G^\omega = \emptyset.
\end{align*} 

However, if we add an atom $a$ to $G$, $$G_a := G\cup \{a\},$$ then
\begin{align*} 
	b\in G_a^\omega \quad\Leftrightarrow\quad \text{$b$ is reachable from $a$ in $G$} \quad\xLeftrightarrow{\text{\prettyref{f:reachability}}}\quad b \leftarrow a\in G^+.
\end{align*} More generally speaking, if we define, for an interpretation $I$,
\begin{align*} 
	G_I := G\cup I,
\end{align*} then
\begin{align*} 
	b\in G_I^\omega \quad\Leftrightarrow\quad \text{$b$ is reachable in $G$ from some vertex in $I$}.
\end{align*}

\begin{definition} We say that two graphs $G$ and $H$ are ${}^\omega$-$I$-equivalent if $G_I\equiv_\omega H_I$.
\end{definition}

\begin{remark} Notice that two graphs are uniformly equivalent iff they are ${}^\omega$-$I$-equivalent, for every $I$.
\end{remark}

\todo[inline]{When are two graphs uniformly equivalent? To what algebraic graph property does uniform equivalence correspond?}

\subsection{Strong equivalence}

...

\section{Krohn-Rhodes-like decomposition}\label{sec:Krohn-Rhodes}

... Krohn-Rhodes-like decomposition \cite{Krohn65}:

\begin{conjecture} Every Krom program can be sequentially decomposed into aperiodic and permutation programs.
\end{conjecture}

\begin{example} 
\todo[inline]{}
\end{example}

\todo[inline]{Kann ich auch Start- und Endzustaende mit Krom-Programmen verarbeiten?}

\section{Green's relations}\label{sec:Green}

In this section, we study Green's fundamental $\mathcal{L,R,J}$-relations \cite{Green51} within finite Krom monoids.

\todo[inline]{Fuege Resultate aus \cite{Antic23-3} ein}

\subsection{\texorpdfstring{$\mathcal L$-Relation}{}}

\begin{definition} 
\todo[inline]{}
\end{definition}

\begin{example} 
\todo[inline]{}
\end{example}

The next result shows that Green's $\mathcal L$-relation is trivial on interpretations:

\begin{proposition} All interpretations are $\mathcal L$-equivalent.
\end{proposition}
\begin{proof} A direct consequence of the fact that interpretations are left zeros \prettyref{eq:IK=I}.
\end{proof}

\begin{proposition} For any Krom program $K$ and interpretation $I$,
\begin{align*} 
	I\leqq_ \mathcal L K
\end{align*} and
\begin{align*} 
	K\equiv_ \mathcal L I \quad\Leftrightarrow\quad K\in \mathbb I.
\end{align*}
\end{proposition}
\begin{proof} The first relation is a direct consequence of the fact that interpretations are left zeros \prettyref{eq:IK=I} and the second follows from the fact that $I = LK$, for some Krom program $L$, iff $K$ is an interpretation.
\end{proof}

\begin{proposition} For any Krom program $K$ and interpretation $I$,
\begin{align*} 
	K\cup I\leqq_ \mathcal L K.
\end{align*} Hence,
\begin{align*} 
	K\leqq_ \mathcal L p(K).
\end{align*}
\end{proposition}
\begin{proof} A direct consequence of \prettyref{eq:K_cup_I=I^astK}.
\end{proof}

\begin{lemma} For any Krom programs $K$ and $L$, if $K\leqq_ \mathcal L L$ then $K = ML$, for some Krom program $M$, and thus
\begin{align*} 
	f(K) &= f(ML) \stackrel{\ref{eq:f(KL)}}= f(M)\cup p(M)f(L)\\
	p(K) &= p(ML) \stackrel{\ref{eq:p(KL)}}= p(M)p(L)\\
	h(K) &= h(ML) \stackrel{\ref{eq:}}\subseteq h(M)\\
	b(K) &= b(ML) \stackrel{\ref{eq:}}\subseteq b(L)
\end{align*} and
\begin{align*} 
	f(M)\subseteq f(K).
\end{align*}
\end{lemma}
\begin{proof} 
\todo[inline]{}
\end{proof}

\begin{theorem} $K\equiv_ \mathcal L L$ iff 
\todo[inline]{}
\end{theorem}
\begin{proof} 
\todo[inline]{}
\end{proof}

\subsection{\texorpdfstring{$\mathcal R$-Relation}{}}

\begin{definition} 
\todo[inline]{}
\end{definition}

\begin{example} 
\todo[inline]{}
\end{example}

\begin{lemma} For any Krom programs $K$ and $L$, if
\begin{align*} 
	K\equiv_ \mathcal R L,
\end{align*} then
\begin{align*} 
	f(K) &= f(L)\\
	h(K) &= h(L).
\end{align*}
\end{lemma}
\begin{proof} 
\todo[inline]{}
\end{proof}

\begin{proposition} For any Krom programs $K$ and $L$, we have
\begin{align*} 
	K\equiv_ \mathcal R L \quad\Leftrightarrow\quad p(K)\equiv_ \mathcal R p(L).
\end{align*}
\end{proposition}
\begin{proof} 
\todo[inline]{}
\end{proof}

\begin{theorem} 
\todo[inline]{}
\end{theorem}
\begin{proof} 
\todo[inline]{}
\end{proof}

\subsection{\texorpdfstring{$\mathcal J$-Relation}{}}

\begin{definition} 
\todo[inline]{}
\end{definition}

\begin{example} 
\todo[inline]{}
\end{example}

\begin{theorem} 
\todo[inline]{}
\end{theorem}
\begin{proof} 
\todo[inline]{}
\end{proof}

\section{Future work}\label{sec:FW}

\subsection{}

A natural next step is to consider \emph{first-order} Krom logic programs containing predicate, function, and constant symbols. This includes some well-known programs from the literature like the one for the addition of numerals and appending of lists. The so-obtained monoids are no longer finite and include, via the grounding of rules (i.e., replacing variables with terms not containing variables), infinite propositional programs.

Another line of research is to try to generalize or adapt results from this paper from Krom to arbitrary propositional,  first-order, and answer set programs containing negation as failure \cite{Clark78}. The former is challenging since the sequential composition of arbitrary propositional programs is in general not associative and does not distribute over union (cf. \cite[Example 10]{Antic21-1}) which means that most of the results of this paper are not directly transferable. The latter is difficult since the sequential composition of answer set programs is rather complicated (cf. \cite{Antic21-2}).

\subsection{Graphs}

In the future, we wish to study the connections between propositional Krom logic programs and graph theory, briefly analyzed in \prettyref{sec:graphs}, more deeply.


Since the space of all Krom programs forms a monoid with respect to sequential composition, and almost\footnote{Recall that the only reason for the space of all Krom programs not to form a semiring is that in case $K$ contains facts, $K\emptyset \stackrel{\ref{eq:f(K)=K0}}= f(K)\neq \emptyset$ thus violating the semiring axiom $a\cdot 0 = 0$.} a semiring with respect to composition and union, we can ask all kinds of algebraic questions. A particularly interesting one seems to be: What are the ``prime'' propositional Krom logic programs? Computations suggest the following Krohn-Rhodes-like conjecture \cite{Krohn65}:

\begin{conjecture}\label{cj:KR} Every propositional Krom logic program can be sequentially decomposed into aperiodic and permutation programs.
\todo[inline]{Vielleicht ein Spezialfall, der sich beweisen lässt — z.B. für kleine Alphabete oder spezielle Programmklassen}
\todo[inline]{Eine klarere Verbindung zum klassischen Krohn-Rhodes-Theorem, also warum die Analogie strukturell sinnvoll ist}
\end{conjecture}

\begin{example} 
\end{example}

\begin{example} 
\end{example}

\subsection{Analogical proportions}

Another interesting line of future research is to study analogical proportions in finite Krom monoids and omegaseminearrings.

\todo[inline]{\cite{Antic22,Antic23-22,Antic22-4}, \cite{Antic21-3,Antic22-2}, \cite{Antic23-23}}

\section{Depends on}

\begin{definition} For any Krom programs $K$ and $L$, we say that $K$ \emph{depends on} $L$ if $b(K)\cap h(L)\neq \emptyset$.
\end{definition}

\begin{proposition} For any Krom programs $K$ and $L$, if $K$ does not depend on $L$, then
\begin{align*} 
	(K\cup L)^\omega = (K^\omega\cup L)^\omega.
\end{align*}
\end{proposition}
\begin{proof} 
\todo[inline]{}
\end{proof}

\begin{corollary} For any Krom programs $K$ and $L$, if $K$ does not depend on $L$, then
\begin{align*} 
	(K\cup L)^\omega = ...
\end{align*}
\todo[inline]{verwende die Formel fuer $(K\cup I)^\omega$}
\end{corollary}

\newpage
\section*{TODOs}

\fi
\end{document}